\documentclass[runningheads]{llncs}
\usepackage{amssymb,amsmath}
\usepackage{graphics,graphicx,color}
\usepackage{enumerate,paralist,hyperref}
\usepackage[labelsep=period,labelfont=bf,size=small,compatibility=false]{caption}
\usepackage{subfig,xspace}
\usepackage{setspace,xstring}

\newcommand{\ver}{arxiv}

\newcommand{\appa}[1]{%
    \IfEqCase{#1}{%
        {conf}{\cite{arxiv}\xspace}%
        {arxiv}{Appendix~\ref{app:bound}\xspace}%
    }[\PackageError{appa}{Undefined option to app: #1}{}]%
}%

\newcommand{\appb}[1]{%
    \IfEqCase{#1}{%
        {conf}{\cite{arxiv}\xspace}%
        {arxiv}{Appendix~\ref{app:proofs}\xspace}%
    }[\PackageError{appb}{Undefined option to app: #1}{}]%
}%

\newcommand{\fs}{\mathcal{F}_s}
\newcommand{\rephrase}[2]{\noindent\textbf{#1}.~\emph{#2}}

\begin{document}
\title{On the Density of non-Simple 3-Planar Graphs\thanks{This work has been supported by DFG grant Ka812/17-1.}}
\author{Michael~A.~Bekos\inst{1}, Michael~Kaufmann\inst{1}, Chrysanthi~N.~Raftopoulou\inst{2}} 
\institute{Institut f\"ur Informatik, Universit\"at T\"ubingen, T\"ubingen, Germany\\
\texttt{$\{$bekos,mk$\}$@informatik.uni-tuebingen.de}
\and
School of Applied Mathematics \& Physical Sciences, NTUA, Athens, Greece\\
\texttt{crisraft@mail.ntua.gr}}


\maketitle

\begin{abstract}
A \emph{$k$-planar graph} is a graph that can be drawn in the plane
such that every edge is crossed at most $k$ times. For $k \leq 4$,
Pach and T\'oth~\cite{PachT97} proved a bound of $(k+3)(n-2)$ on the
total number of edges of a $k$-planar graph, which is tight for
$k=1,2$. For $k=3$, the bound of $6n-12$ has been improved to
$\frac{11}{2}n-11$ in~\cite{Pach2006} and has been shown to be
optimal up to an additive constant for simple graphs. In this paper,
we prove that the bound of $\frac{11}{2}n-11$ edges also holds for
non-simple $3$-planar graphs that admit drawings in which non-homotopic parallel edges
and self-loops are allowed. Based on this result, a characterization
of \emph{optimal $3$-planar graphs} (that is, $3$-planar graphs with
$n$ vertices and exactly $\frac{11}{2}n-11$ edges) might be
possible, as to the best of our knowledge the densest known simple
$3$-planar is not known to be optimal.
\end{abstract}

\section{Introduction}
\label{sec:introduction}
Planar graphs play an important role in graph drawing and
visualization, as the avoidance of crossings and occlusions is
central objective in almost all applications
\cite{BETT99,KW1999}. The theory of planar graphs
\cite{Harary91} could be very nicely applied and used for
developing great layout algorithms \cite{FPP90,Tamassia87,Tutte63}
based on the planarity concepts. Unfortunately, real-world graphs are
usually not planar despite of their sparsity. With this background,
an initiative has formed in recent years to develop a suitable theory
for \emph{nearly planar graphs}, that is, graphs with various
restrictions on their crossings, such as limitations on the number
of crossings per edge (e.g., $k$-planar graphs~\cite{Ringel65}), 
avoidance of local crossing configurations (e.g., quasi planar
graphs~\cite{AAPPS97}, fan-crossing free graphs~\cite{CHKK13},
fan-planar graphs~\cite{KU14}) or restrictions on the crossing angles
(e.g., RAC graphs~\cite{DEL11}, LAC graphs~\cite{DGMW11}). For
precise definitions, we refer to the literature mentioned above.

The most prominent is clearly the concept of $k$-planar graphs,
namely graphs that allow drawings in the plane such that each edge is
crossed at most $k$ times by other edges. The simplest case $k=1$,
i.e., $1$-planar graphs~\cite{Ringel65}, has been subject of
intensive research in the past and it is quite well understood, see
e.g.~\cite{BB0R15,Borodin95,Brandenburg14,BEGGHR12,GB07,PachT97}.
For $k \geq 2$, the picture is much less clear. Only few papers on
special cases appeared, see e.g., \cite{ABGH12,HN15}.

Pach and T\'oth's paper~\cite{PachT97} stands out and contributed a
lot to the understanding of nearly planar graphs. The paper considers
the number of edges in simple $k$-planar graphs for general $k$. Note
the well-known bound of $3n-6$ edges for planar graphs deducible from
Euler's formula. For small $k = 1,2,3$ and $4$, bounds of $4n-8$,
$5n-10$, $6n-12$ and $7n-14$ respectively, are proven which are~tight
for $k =1$ and $k=2$. This sequence seems to suggest a bound of
$O(kn)$ for general $k$, but Pach and T\'oth also gave an upper bound
of $4.1208 \sqrt k n$. Unfortunately, this bound is still quite large
even for medium $k$ (for $k=9$, it gives $12.36 n$). Meanwhile for
$k=3$ and $k=4$, the bounds above have been improved to $5.5n-11$ and
$6n-12$ in~\cite{Pach2006} and \cite{Ackerman15}, respectively. In
this paper, we prove that the bound on the number of edges for $k=3$
also holds for non-simple $3$-planar graphs that do not contain
homotopic parallel edges and homotopic self-loops. Our extension
required substantially different approaches and relies more on
geometric techniques than the more combinatorial ones given in
\cite{Pach2006} and \cite{Ackerman15}. We believe that it might also
be central for the characterization of \emph{optimal} $3$-planar
graphs (that is, $3$-planar graphs with $n$ vertices and exactly
$\frac{11}{2}n-11$ edges), since the densest known simple $3$-planar
graph has only $\frac{11n}{2}-15$ edges and does not reach the known
bound.

The remaining of this paper is structured as follows: Some
definitions and preliminaries are given in
Section~\ref{sec:preliminaries}. 
In Sections~\ref{sec:upper-bound} and~\ref{sec:density}, we give
significant insights in structural properties of $3$-planar graphs in
order to prove that $3$-planar graphs on $n$ vertices cannot have
more than $\frac{11}{2}n-11$ edges. We conclude in
Section~\ref{sec:conclusions} with open problems.

\section{Preliminaries}
\label{sec:preliminaries}

A \emph{drawing} of a graph $G$ is a representation of $G$ in the
plane, where the vertices of $G$ are represented by distinct points
and its edges by Jordan curves joining the corresponding pairs of
points, so that:%
\begin{inparaenum}[(i)]
\item no edge passes through a vertex different from its endpoints,
\item no edge crosses itself and
\item no two edges meet tangentially.
\end{inparaenum}
In the case where $G$ has multi-edges, we will further assume that
both the bounded and the unbounded closed regions defined by any
pair of self-loops or parallel edges of $G$ contain at least one
vertex of $G$ in their interior. Hence, the drawing of $G$ has no
\emph{homotopic} edges. In the following when referring to
$3$-planar graphs we will mean that non-homotopic edges are allowed in the corresponding drawings.
We call such graphs \emph{non-simple}.

Following standard naming conventions, we refer to a $3$-planar graph
with $n$ vertices and maximum possible number of edges as
\emph{optimal $3$-planar}. Let $H$ be an optimal $3$-planar graph on
$n$ vertices together with a corresponding $3$-planar drawing
$\Gamma(H)$. Let also $H_p$ be a subgraph of $H$ with the largest
number of edges, such that in the drawing of $H_p$ (that is inherited
from $\Gamma(H)$) no two edges cross each other. We call $H_p$ a
\emph{maximal planar substructure} of $H$. Among all possible optimal
$3$-planar graphs on $n$ vertices, let $G=(V,E)$ be the one with the
following two properties:%
\begin{inparaenum}[(a)]
\item \label{p:1} its maximal planar substructure, say $G_p=(V,E_p)$,
has maximum number of edges among all possible planar substructures
of all optimal $3$-planar graphs,
\item \label{p:2} the number of crossings in the drawing of $G$ is
minimized over all optimal $3$-planar graphs subject to~(\ref*{p:1}).
\end{inparaenum} 
We refer to $G$ as \emph{crossing-minimal optimal $3$-planar graph}.

With slight abuse of notation, let $G-G_p$ be obtained from $G$ by
removing only the edges of $G_p$ and let $e$ be an edge of $G-G_p$.
Since $G_p$ is maximal, edge $e$ must cross at least one edge of
$G_p$. We refer to the part of $e$ between an endpoint of $e$ and the
nearest crossing with an edge of $G_p$ as \emph{stick}. The parts of
$e$ between two consecutive crossings with $G_p$ are called
\emph{middle parts}. Clearly, $e$ consists of exactly $2$ sticks and
$0$, $1$, or $2$ middle parts. A stick of $e$ lies completely in a
face of $G_p$ and crosses at most two other edges of $G-G_p$ and an
edge of this particular face. A stick of $e$ is called \emph{short},
if there is a walk along the face boundary from the endpoint of the
stick to the nearest crossing point with $G_p$, which contains only
one other vertex of the face boundary. Otherwise, the stick of $e$ is
called \emph{long};  see Figure~\ref{fig:non_simple}. A middle part of
$e$ also lies in a face of $G_p$. We say that $e$ \emph{passes
through} a face of $G_p$, if there exists a middle part of $e$ that
completely lies in the interior of this particular face. We refer to
a middle part of an edge that crosses consecutive edges of a face of
$G_p$ as \emph{short middle part}. Otherwise, we call it \emph{far
middle part}.

\begin{figure}[t]
 \centering
    \begin{minipage}[b]{.19\textwidth}
        \centering
        \subfloat[\label{fig:non_simple}{}]
        {\includegraphics[width=\textwidth,page=1]{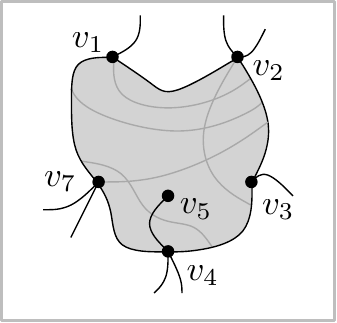}}
    \end{minipage}
	 \begin{minipage}[b]{.19\textwidth}
		\centering
		\subfloat[\label{fig:association}{}]
		{\includegraphics[width=\textwidth,page=1]{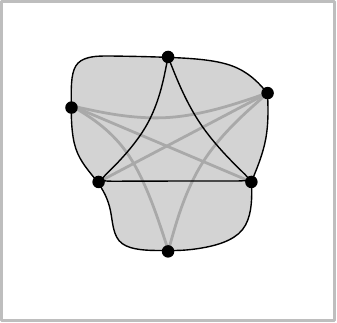}}
	 \end{minipage}
 \caption{
 (a)~Illustration of a non-simple face $\{v_1,v_2,\ldots,v_7\}$; $v_6$ is identified with $v_4$. 
 The sticks from $v_1$ and $v_2$ are short, while the one from $v_7$ is long. 
 All other edge segments are middle-parts.
 (b) The case, where two triangles of type $(3,0,0)$ are associated to the same triangle.}
 \label{fig:general}
\end{figure}

Let $\fs=\{v_1,v_2,\ldots,v_s\}$ be a face of $G_p$ with $s \geq
3$. The order of the vertices (and subsequently the order of the
edges) of $\fs$ is determined by a walk around the boundary of $\fs$
in clockwise direction. Since $\fs$ is not necessarily simple, a
vertex (or an edge, respectively) may appear more than once in this
order; see Figure~\ref{fig:non_simple}.  We say that $\fs$ is of type
$(\tau_1,\tau_2,\ldots,\tau_s)$ if for each $i=1,2,\ldots,s$ vertex
$v_i$ is incident to $\tau_i$ sticks of $\fs$ that lie between
$(v_{i-1},v_i)$ and $(v_i,v_{i+1})$\footnote{In the remainder of the
paper, all indices are subject to $(mod~s)+1$.}.

\begin{lemma}[Pach and T\'oth~\cite{PachT97}]
A triangular face of $G_p$ contains at most $3$ sticks.
\label{lem:no_of_sticks}
\end{lemma}
\begin{proof}
Consider a triangular face $\mathcal{T}$ of $G_p$ of type
$(\tau_1,\tau_2,\tau_3)$. Clearly, $\tau_1,\tau_2,\tau_3 \leq 3$, as
otherwise an edge of $G_p$ has more than three crossings. Since a
stick of $\mathcal{T}$ cannot cross more than two other sticks of
$\mathcal{T}$, it follows that $\tau_1+\tau_2+\tau_3 \leq 3$.\qed
\end{proof}

\section{The Density of non-Simple 3-Planar Graphs}
\label{sec:upper-bound}
Let $G=(V,E)$ be a crossing-minimal optimal $3$-planar graph with $n$
vertices drawn in the plane. Let also $G_p=(V,E_p)$ be the maximal
planar substructure of $G$. In this section, we will prove that $G$
cannot have more than $\frac{11n}{2}-11$ edges, assuming that $G_p$
is fully triangulated, i.e., $|E_p| = 3n-6$. This assumption will
be proved in Section~\ref{sec:density}. Next, we prove that the
number of triangular faces of $G_p$ with exactly $3$ sticks cannot be
larger than those with at most $2$ sticks.

\begin{lemma}
We can uniquely associate each triangular face of $G_p$ with $3$
sticks to a neighboring triangular face of $G_p$ with at most $2$
sticks.
\label{lem:associate}
\end{lemma}
\begin{proof}
Let $\mathcal{T}=\{v_1,v_2,v_3\}$ be a triangular face of $G_p$. By
Lemma~\ref{lem:no_of_sticks}, we have to consider three types for
$\mathcal{T}$: $(3,0,0)$, $(2,1,0)$ and $(1,1,1)$.
\begin{itemize}
\item \emph{$\mathcal{T}$ is of type $(3,0,0)$:}
Since $v_1$ is incident to $3$ sticks of $\mathcal{T}$, edge
$(v_2,v_3)$ is crossed three times. Let $\mathcal{T}'$ be the
triangular face of $G_p$ neighboring $\mathcal{T}$ along
$(v_2,v_3)$. We have to consider two cases:%
\begin{inparaenum}[(a)]
\item \label{a:c1} one of the sticks of $\mathcal{T}$ ends at a corner of $\mathcal{T}'$, and
\item \label{a:c2} none of the sticks of $\mathcal{T}$ ends at a corner of $\mathcal{T}'$.
\end{inparaenum}
In Case~(\ref*{a:c1}), the two remaining sticks of $\mathcal{T}$
might use the same or different sides of $\mathcal{T}'$ to exit it.
In both subcases, it is not difficult to see that $\mathcal{T}'$ can
have at most two sticks. In Case~(\ref*{a:c2}), we
again have to consider two subcases, depending on whether all sticks
of $\mathcal{T}$ use the same side of $\mathcal{T}'$ to pass through
it or two different ones. In the former case, it is not difficult to
see that $\mathcal{T}'$ cannot have any stick, while in the later
$\mathcal{T}'$ can have at most one stick. In all aforementioned
cases, we associate $\mathcal{T}$ with $\mathcal{T}'$.

\item \emph{$\mathcal{T}$ is of type $(2,1,0)$:}
Since $v_2$ is incident to one stick of $\mathcal{T}$, edge
$(v_1,v_3)$ is crossed at least once. We associate $\mathcal{T}$
with the triangular face $\mathcal{T}'$ of $G_p$ neighboring
$\mathcal{T}$ along $(v_1,v_3)$. Since the stick of $\mathcal{T}$
that is incident to $v_2$ has three crossings in $\mathcal{T}$,
$\mathcal{T}'$ has no sticks emanating from $v_1$ or $v_3$. In
particular, $\mathcal{T}'$ can have at most one additional stick
emanating from its third vertex.

\item \emph{$\mathcal{T}$ is of type $(1,1,1)$:} This actually 
cannot occur. Indeed,
if $\mathcal{T}$ is of type $(1,1,1)$, then all sticks of
$\mathcal{T}$ have already three crossings each. Hence, the three
triangular faces adjacent to $\mathcal{T}$ define a $6$-gon in
$G_p$, which contains only six interior edges. So, we
can easily remove them and replace them with $8$ interior edges
(see, e.g., Figure~\ref{fig:association}), contradicting thus the
optimality of~$G$.
\end{itemize}
Note that our analysis also holds for non-simple triangular faces.
We now show that the assignment is unique. This
holds for triangular faces of type $(2,1,0)$, since a triangular face
that is associated with one of type $(2,1,0)$ cannot contain two
sides each with two crossings, which implies that it cannot be
associated with another triangular face with three sticks. This
leaves only the case that two $(3,0,0)$ triangles are associated with
the same triangle $\mathcal{T}'$ (see, e.g., the triangle with the
gray-colored edges in Figure~\ref{fig:association}). In this case,
there exists another triangular face (bottommost in
Figure~\ref{fig:association}), which has exactly two sticks because
of $3$-planarity. In addition, this face cannot be associated with
some other triangular face. Hence, one of the two type-$(3,0,0)$
triangular faces associated with $\mathcal{T}'$ can be assigned to
this triangular face instead resolving the conflict.\qed
\end{proof}
\noindent We are now ready to prove the main theorem of this section.

\begin{theorem}
A $3$-planar graph of $n$ vertices has at most $\frac{11}{2}n -11$
edges, which is a tight bound.
\label{thm:main}
\end{theorem}
\begin{proof}
Let $t_i$ be the number of triangular faces of $G_p$ with exactly $i$
sticks, $0 \leq i \leq 3$. The argument starts by counting
the number of triangular faces of $G_p$ with exactly $3$ sticks. From
Lemma~\ref{lem:associate}, we conclude that the number $t_3$ of
triangular faces of $G_p$ with exactly $3$ sticks is at most as large
as the number of triangular faces of $G_p$ with $0$, $1$ or $2$
sticks. Hence $t_3 \leq t_0 + t_1 + t_2$. We conclude that
$t_3 \leq t_p/2$, where $t_p$ denotes the number of triangular faces
in $G_p$, since $t_0 + t_1 + t_2 + t_3 = t_p$. Note that by Euler's
formula $t_p = 2n-4$. Hence, $t_3 \leq n-2$. Thus, we have: 
$|E| - |E_p| = (t_1 + 2t_2 + 3t_3)/2 = (t_1 + t_2 + t_3) + (t_3 -
t_1)/2 = (t_p - t_0) + (t_3 - t_1)/2 \leq t_p + t_3/2 \leq 5t_p/4$. 
So, the total number of edges of $G$ is at most: 
$|E| \leq |E_p| + 5t_p/4 \leq 3n - 6 + 5(2n - 4)/4 = 11n/2 -
11$. In~\appa{\ver} we prove that our bound is tight by a
construction similar to the one of Pach et al.~\cite{Pach2006}.\qed
\end{proof}

\section{The Density of the Planar Substructure}
\label{sec:density}
Let $G=(V,E)$ be a crossing-minimal optimal $3$-planar graph with $n$
vertices drawn in the plane. Let also $G_p=(V,E_p)$ be the maximal
planar substructure of $G$. In this section, we will prove that $G_p$
is fully triangulated, i.e., $|E_p| = 3n-6$ (see
Theorem~\ref{thm:triangulated}). To do so, we will explore several
structural properties of $G_p$ (see
Lemmas~\ref{lem:stick_cross}-\ref{lem:two_stick}), assuming that
$G_p$ has at least one non-triangular face, say $\fs=\{v_1, v_2,
\ldots, v_s\}$ with $s \geq 4$. In the first observations, we do not
require that $G_p$ is connected. This is proved in
Lemma~\ref{lem:connected}. Recall that in general $\fs$ is not
necessarily simple, which means that a vertex may appear more than
once along $\fs$. Our goal is to contradict either the
\emph{optimality} of~$G$ (that is, the fact that $G$ contains the
maximum number of edges among all $3$-planar graphs with $n$
vertices) or the \emph{maximality} of~$G_p$ (that is, the fact that
$G_p$ has the maximum number of edges among all planar substructures
of all optimal $3$-planar graphs with $n$ vertices) or the
\emph{crossing minimality} of~$G$ (that is, the fact that $G$ has the
minimum number of crossings subject to the size of the planar
substructure).

\newcommand{\stickcross}{Let $\fs=\{v_1, v_2, \ldots, v_s\}$, $s
\geq 4$ be a non-triangular face of $G_p$. Then, each stick of $\fs$
is crossed at least once within $\fs$.}

\begin{lemma}
\stickcross
\label{lem:stick_cross}
\end{lemma}
\begin{proof}[Sketch]
Assume to the contrary that there exists a stick of $\fs$ that is
not crossed within $\fs$. W.l.o.g.~let $(v_1,v_1')$ be the edge
containing this stick and assume that $(v_1,v_1')$ emanates from
vertex $v_1$ and leads to vertex $v_1'$ by crossing the edge
$(v_i,v_{i+1})$ of $\fs$. We initially prove that $i+1=s$. Next, we
show that there exist two edges $e_1$ and $e_2$ which cross
$(v_i,v_{i+1})$ and are not sticks emanating from $v_1$. The desired
contradiction follows from the observation that we can remove edges
$e_1$, $e_2$ and $(v_1,v_1')$ from $G$ and replace them with the
chord $(v_1,v_{s-1})$ and two additional edges that are both sticks
either at $v_1$ or at $v_s$. In this way, a new graph is obtained,
whose maximal planar substructure has more edges than $G_p$, which
contradicts the maximality of $G_p$. The detailed proof is given 
in~\appb{\ver}.\qed
\end{proof}

\newcommand{\ptnofar}{Let $\fs=\{v_1, v_2, \ldots, v_s\}$, $s \geq
4$ be a non-triangular face of $G_p$. Then, each middle part of
$\fs$ is short, i.e., it crosses consecutive edges of $\fs$.}

\begin{lemma}
\ptnofar
\label{lem:pt_no_far}
\end{lemma}
\begin{proof}[Sketch]
For a proof by contradiction, assume that $(u,u')$ is an edge that
defines a middle part of $\fs$ which crosses two non-consecutive
edges of $\fs$, say w.l.o.g.~$(v_1,v_2)$ and $(v_i,v_{i+1})$, where
$i \neq 2$ and $ i+1 \neq s$. We distinguish two main cases. Either
$(u,u')$ is not involved in crossings in the interior of $\fs$ or
$(u,u')$ is crossed by an edge, say $e$, within $\fs$. In both cases,
it is possible to lead to a contradiction to the maximality of $G_p$;
refer to~\appb{\ver} for more~details.\qed
\end{proof}

\begin{lemma}
Let $\fs=\{v_1, v_2, \ldots, v_s\}$, $s \geq
4$ be a non-triangular face of $G_p$. Then, each stick of
$\mathcal{F}_s$ is short.
\label{lem:stick_no_far}
\end{lemma}
\begin{proof}
Assume for a contradiction that there exists a far stick. Let
w.l.o.g.~$(v_1,v_1')$ be the edge containing this stick and assume
that $(v_1,v_1')$ emanates from vertex $v_1$ and leads to vertex
$v_1'$ by crossing the edge $(v_i,v_{i+1})$ of $\fs$, where $i \neq
2$ and $i+1 \neq s$. If we can replace $(v_1,v'_1)$ either with chord
$(v_1,v_i)$ or with chord $(v_1,v_{i+1})$, then the maximal planar
substructure of the derived graph would have more edges than $G_p$;
contradicting the maximality of $G_p$. Thus, there exist two edges,
say $e_1$ and $e_2$, that cross $(v_i,v_{i+1})$ to the left and to
the right of $(v_1,v'_1)$, respectively; see
Figure~\ref{fig:stick_no_far_1}. By Lemma~\ref{lem:stick_cross}, edge
$(v_1,v'_1)$ is crossed by at least one other edge, say $e$, inside
$\mathcal{F}_s$. Note that by $3$-planarity edge $(v_1,v_1')$ might
also be crossed by a second edge, say $e'$, inside $\fs$.
Suppose first, that $(v_1,v_1')$ has a single crossing inside
$\mathcal{F}_s$. To cope with this case, we propose two
alternatives:%
\begin{inparaenum}[(a)]
\item replace $e_1$ with chord $(v_1,v_{i+1})$ and make vertex
$v_{i+1}$ an endpoint of $e$, or
\item  replace $e_2$ with chord $(v_1,v_i)$ and make vertex $v_i$ an
endpoint of both $e$; see Figures~\ref{fig:stick_no_far_2}
and~\ref{fig:stick_no_far_3}, respectively.
\end{inparaenum}
Since $e$ and $(v_i,v_{i+1})$ are not homotopic, it follows that at
least one of the two alternatives can be applied, contradicting the
maximality of $G_p$.

\begin{figure}[t]
    \centering
    \begin{minipage}[b]{.19\textwidth}
        \centering
        \subfloat[\label{fig:stick_no_far_1}{}]
        {\includegraphics[width=\textwidth,page=1]{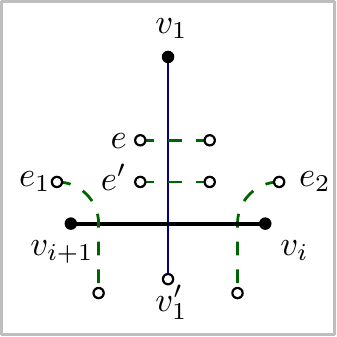}}
    \end{minipage}
    \begin{minipage}[b]{.19\textwidth}
        \centering
        \subfloat[\label{fig:stick_no_far_2}{}]
        {\includegraphics[width=\textwidth,page=2]{images/stick_no_far}}
    \end{minipage} 
	\begin{minipage}[b]{.19\textwidth}
        \centering
        \subfloat[\label{fig:stick_no_far_3}{}]
        {\includegraphics[width=\textwidth,page=3]{images/stick_no_far}}
    \end{minipage}
	\begin{minipage}[b]{.19\textwidth}
        \centering
        \subfloat[\label{fig:stick_no_far_4}{}]
        {\includegraphics[width=\textwidth,page=4]{images/stick_no_far}}
    \end{minipage}	
    	
	\begin{minipage}[b]{.19\textwidth}
        \centering
        \subfloat[\label{fig:stick_no_far_5}{}]
        {\includegraphics[width=\textwidth,page=5]{images/stick_no_far}}
    \end{minipage}
	\begin{minipage}[b]{.19\textwidth}
        \centering
        \subfloat[\label{fig:stick_no_far_6}{}]
        {\includegraphics[width=\textwidth,page=6]{images/stick_no_far}}
    \end{minipage}
	\begin{minipage}[b]{.19\textwidth}
        \centering
        \subfloat[\label{fig:stick_no_far_7}{}]
        {\includegraphics[width=\textwidth,page=7]{images/stick_no_far}}
    \end{minipage}
    \caption{%
    Different configurations used in the proof of Lemma~\ref{lem:stick_no_far}.}
    \label{fig:stick_no_far}
\end{figure}

Consider now the case where $(v_1,v_1')$ has two crossings inside
$\mathcal{F}_s$, with edges $e$ and $e'$. Similarly to the previous
case, we propose two alternatives:%
\begin{inparaenum}[(a)]
\item replace $e_1$ with chord $(v_1,v_{i+1})$ and make vertex
$v_{i+1}$ an endpoint of both $e$ and $e'$, or
\item  replace $e_2$ with chord $(v_1,v_i)$ and make vertex $v_i$ an
endpoint of both $e$ and $e'$; see Figures~\ref{fig:stick_no_far_4}
and~\ref{fig:stick_no_far_5}, respectively.
\end{inparaenum}
Note that in both alternatives the maximal planar substructure of the
derived graph has more edges than $G_p$, contradicting the maximality
of $G_p$. Since $e$ and $e'$ are not homotopic, it follows that one
of the two alternatives is always applicable, as long as, $e$ and
$e'$ are not simultaneously sticks from $v_i$ and $v_{i+1}$,
respectively; see Figure~\ref{fig:stick_no_far_6}. In this scenario,
both alternatives would lead to a situation, where $(v_i,v_{i+1})$
has two homotopic copies. To cope with this case, we observe that
$e$, $e'$ and $(v_1,v'_1)$ are three mutually crossing edges inside
$\fs$. We proceed by removing from $G$ edges $e_1$ and $e_2$, which
we replace by $(v_1,v_i)$ and $(v_1,v_{i+1})$; see
Figure~\ref{fig:stick_no_far_7}. In the derived graph the maximal
planar substructure contains more edges than $G_p$ (in particular,
edges $(v_1,v_i)$ and $(v_1,v_{i+1})$), contradicting its
maximality.\qed
\end{proof}

\begin{lemma}
The planar substructure $G_p$ of a crossing-minimal optimal
$3$-planar graph $G$ is connected.
\label{lem:connected}
\end{lemma}
\begin{proof}
Assume to the contrary that the maximum planar substructure $G_p$ of
$G$ is not connected and let $G_p'$ be a connected component of
$G_p$. Since $G$ is connected, there is an edge of $G-G_p$ that
bridges $G_p'$ with $G_p-G_p'$. By definition, this edge is either a
stick or a passing through edge for the common face of $G_p'$ and
$G-G_p'$. In both cases, it has to be short (by
Lemmas~\ref{lem:pt_no_far} and~\ref{lem:stick_no_far}); a
contradiction.\qed
\end{proof}

\noindent  In the next two lemmas, we consider the case where a
non-triangular face $\fs=\{v_1, v_2, \ldots, v_s\}$, $s \geq 4$ of
$G_p$ has no sticks. Let $br(\fs)$ and $\overline{br}(\fs)$ be the
set of bridges and non-bridges of $\fs$, respectively (in 
Figure~\ref{fig:non_simple}, edge $(v_4,v_5)$ is a bridge). In the absence
of sticks, a passing through edge of $\fs$ \emph{originates} from one
of its end-vertices, crosses an edge of $\overline{br}(\fs)$ to
\emph{enter} $\fs$, passes through $\fs$ (possibly by defining two
middle parts, if it crosses an edge of $br(\fs)$), crosses another
edge of $\overline{br}(\fs)$ to \emph{exit} $\fs$ and
\emph{terminates} to its other end-vertex. We \emph{associate} the
edge of $\overline{br}(\fs)$ that is used by the passing through edge
to enter (exit) $\fs$ with the origin (terminal) of this passing
through edge. Let  $\overline{s_b}$ and $s_b$ be the number of edges
in $\overline{br}(\fs)$ and  $br(\fs)$, respectively.
Let also $\widehat{s_b}$ be the number of edges of
$\overline{br}(\fs)$ that are crossed by no passing through edge of
$\fs$. Clearly, $\widehat{s_b} \leq \overline{s_b}$ and
$s=\overline{s_b} + 2s_b$.

\begin{lemma}
Let $\fs=\{v_1, v_2, \ldots, v_s\}$, $s \geq 4$ be a non-triangular
face of $G_p$ that has no sticks. Then, the number $\widehat{s_b}$ of
non-bridges of $\fs$ that are crossed by no passing through edge of
$\fs$ is strictly less than half the number $\overline{s_b}$ of of
non-bridges of $\fs$, that is, $\widehat{s_b} <
\frac{\overline{s_b}}{2}$.
\label{lem:uncrossed_edges}
\end{lemma}
\begin{proof}
For a proof by contradiction assume that $\widehat{s_b} \geq
\frac{\overline{s_b}}{2}$. Since at most $\frac{\overline{s_b}}{2}$
edges of $\fs$ can be crossed (each of which at most three times) and
each passing through edge of $\fs$ crosses two edges of
$\overline{br}(\fs)$, it follows that $|pt(\fs)| \leq  \lfloor
\frac{3\overline{s_b}}{4} \rfloor$, where $pt(\fs)$ denotes the set
of passing through edges of $\fs$. To obtain a contradiction, we
remove from $G$ all edges that pass through $\fs$ and we introduce
$2s-6$ edges $\{(v_1,v_i):~2<i<s\} \cup \{(v_i,v_i+2):~2 \leq i \leq
s-2\}$ that lie completely in the interior of $\fs$. This simple
operation will lead to a larger graph (and therefore to a
contradiction to the optimality of $G$) or to a graph of the same
size but with larger planar substructure (and therefore to a
contradiction to the maximality of $G_p$) as long as $s > 4$. For $s
= 4$, we need a different argument. By Lemma~\ref{lem:pt_no_far}, we
may assume that all three passing through edges of $\fs$ cross two
consecutive edges of $\fs$, say w.l.o.g.~$(v_1,v_2)$ and $(v_2,v_3)$.
This implies that chord $(v_1,v_3)$ can be safely added to $G$; a
contradiction to the optimality of $G$.\qed
\end{proof}

\newcommand{\stickexist}{Let $\fs=\{v_1, v_2, \ldots, v_s\}$, $s
\geq 4$ be a non-triangular face of $G_p$. Then, $\fs$ has at least
one stick.}

\begin{lemma}
\stickexist
\label{lem:stick_exist}
\end{lemma}
\begin{proof}[Sketch]
For a proof by contradiction, assume that $\fs$ has no sticks. By
Lemma~\ref{lem:uncrossed_edges}, it follows that there exist at
least two incident edges of $\overline{br}(\fs)$ that are crossed by
passing through edges of $\fs$, say w.l.o.g.~$(v_s,v_1)$ and
$(v_1,v_2)$. Note that these two edges are not bridges of $\fs$.
If $s+\widehat{s_b} + 2s_b \geq 6$, then as in the proof of
Lemma~\ref{lem:uncrossed_edges}, it is possible to construct a graph
that is larger than $G$ or of equal size as $G$ but with larger planar
substructure. The same holds when $s+\widehat{s_b} +2s_b = 5$ (that
is, $s=5$ and $\widehat{s_b} = s_b=0$ or $s=4$, $\widehat{s_b}=1$
and $s_b=0$). Both cases, contradict either the
optimality of $G$ or the maximality of $G_p$. The case where
$s+\widehat{s_b}+2s_b=4$ is slightly more involved; refer to 
\appb{\ver}.\qed
\end{proof}

\begin{figure}[t]
    \centering
    \begin{minipage}[b]{.19\textwidth}
        \centering
        \subfloat[\label{fig:stick_no_three_1}{}]
        {\includegraphics[width=\textwidth,page=1]{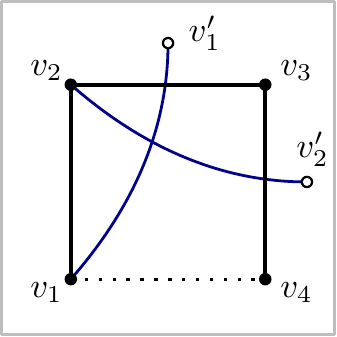}}
    \end{minipage}
		\begin{minipage}[b]{.19\textwidth}
        \centering
        \subfloat[\label{fig:stick_no_three_2}{}]
        {\includegraphics[width=\textwidth,page=2]{images/stick_no_three}}
    \end{minipage}
		\begin{minipage}[b]{.19\textwidth}
        \centering
        \subfloat[\label{fig:stick_no_three_3}{}]
        {\includegraphics[width=\textwidth,page=3]{images/stick_no_three}}
    \end{minipage}
		\begin{minipage}[b]{.19\textwidth}
        \centering
        \subfloat[\label{fig:stick_no_three_4}{}]
        {\includegraphics[width=\textwidth,page=4]{images/stick_no_three}}
    \end{minipage}
		\begin{minipage}[b]{.19\textwidth}
        \centering
        \subfloat[\label{fig:stick_no_three_5}{}]
        {\includegraphics[width=\textwidth,page=5]{images/stick_no_three}}
    \end{minipage}
    \caption{%
    Different configurations used in Lemma~\ref{lem:stick_no_three}.}
    \label{fig:stick_no_three}
\end{figure}

\noindent By Lemma~\ref{lem:stick_no_far}, all sticks of $\fs$ are
short. A stick $(v_i,v'_i)$ of $\fs$ is called \emph{right}, if it
crosses edge $(v_{i+1},v_{i+2})$ of $\fs$. Otherwise, stick
$(v_i,v'_i)$ is called \emph{left}. Two sticks are called
\emph{opposite}, if one is left and the other one is right.

\begin{lemma}
Let $\fs=\{v_1, v_2, \ldots, v_s\}$, $s \geq 4$ be a non-triangular
face of $G_p$. Then, $\fs$ has not three mutually crossing
sticks.
\label{lem:stick_no_three}
\end{lemma}
\begin{proof}
Suppose to the contrary that there exist three mutually crossing
sticks of $\fs$ and let $e_i$, for $i=1,2,3$ be the edges containing
these sticks.  W.l.o.g.~we assume that at least two of them are right
sticks, say $e_1$ and $e_2$. Let $e_1=(v_1,v'_1)$. Then,
$e_2=(v_2,v'_2)$; see Figure~\ref{fig:stick_no_three_1}. Since $e_1$,
$e_2$ and $e_3$ mutually cross, $e_3$ can only contain a left stick.
By Lemma~\ref{lem:stick_no_far} its endpoint on $\fs$ is $v_3$ or
$v_4$. The first case is illustrated in
Figure~\ref{fig:stick_no_three_2}. Observe that $(v_1,v_2)$ of $\fs$
is only crossed by~$e_3$. Indeed, if there was another edge crossing
$(v_1,v_2)$, then it would also cross $e_1$ or $e_2$, both of which
have three crossings. Hence, $e_3$ can be replaced with 
$(v_1,v_3)$; see Figure~\ref{fig:stick_no_three_3}. The maximal
planar substructure of the derived graph would have more edges than
$G_p$, contradicting the maximality of $G_p$. The case where $v_4$ is
the endpoint of $e_3$ on $\fs$ is illustrated in
Figure~\ref{fig:stick_no_three_5}. Suppose that there exists an edge
crossing $(v_2,v_3)$ of $\fs$ to the left of $e_3$. This edge should
also cross $e_2$ or $e_3$, which is not possible since both edges
have three crossings. So, we can replace $e_3$ with chord $(v_2,v_4)$
as in Figure~\ref{fig:stick_no_three_5}, contradicting the maximality
of $G_p$.\qed
\end{proof}

\newcommand{\stickcrossonce}{Let $\fs=\{v_1, v_2, \ldots, v_s\}$, $s
\geq 4$ be a non-triangular face of $G_p$. Then, each stick of $\fs$
is crossed exactly once within $\fs$.}

\begin{lemma}
\stickcrossonce
\label{lem:stick_cross_once}
\end{lemma}
\begin{proof}[Sketch]
The detailed proof is given in \appb{\ver}. By
Lemma~\ref{lem:stick_cross}, each stick of $\fs$ is crossed at least
once within $\fs$. So, the proof is given by contradiction either to
the optimality of $G$ or to the maximality of $G_p$, assuming the
existence of a stick of $\fs$ that is crossed twice within $\fs$, say
by edges $e_1$ and $e_2$. Note that by $3$-planarity a stick of $\fs$
cannot be further crossed within $\fs$. First, we prove that $e_1$
and $e_2$ do not cross each other. Then, we show that $e_1$ and $e_2$
cannot be simultaneously passing through $\fs$. The desired
contradiction is obtained by considering two main cases:
Either $e_1$ passes through $\fs$ (and therefore, $e_2$ is a stick of
$\fs$) or both $e_1$ and $e_2$ are sticks of $\fs$.\qed
\end{proof}

\begin{lemma}
Let $\fs=\{v_1, v_2, \ldots, v_s\}$, $s \geq 4$ be a non-triangular
face of $G_p$. Then, there are no crossings between sticks and middle
parts of $\mathcal{F}_s$.
\label{lem:stick_cross_pt} 
\end{lemma}
\begin{proof}
Assume to the contrary that there exists a stick, say of edge
$(v_1,v'_1)$ that emanates from vertex $v_1$ of $\fs$ (towards
$v_1'$), which is crossed by a middle part of $(u,u')$ of $\fs$. By
Lemma~\ref{lem:stick_cross_once}, this stick cannot have another
crossing within $\fs$.  By Lemma~\ref{lem:stick_no_far}, we can
assume w.l.o.g.~that $(v_1,v_1')$ is a right stick, i.e.,
$(v_1,v_1')$ crosses $(v_2,v_3)$. By Lemma~\ref{lem:pt_no_far}, edge
$(u,u')$ crosses two consecutive edges of $\mathcal{F}_s$. We
distinguish two cases based on whether $(v_1,v_1')$ crosses
$(v_s,v_1)$ and $(v_1,v_2)$ of $\fs$ or $(v_1,v_1')$ crosses
$(v_1,v_2)$ and $(v_2,v_3)$ of $\fs$; see
Figures~\ref{fig:stick_cross_pt_1} and \ref{fig:stick_cross_pt_3}
respectively.

\begin{figure}[t]
	\centering
    \begin{minipage}[b]{.19\textwidth}
        \centering
        \subfloat[\label{fig:stick_cross_pt_1}{}]
        {\includegraphics[width=\textwidth,page=1]{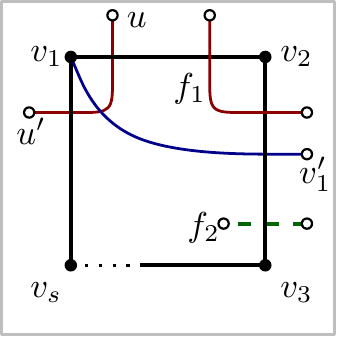}}
    \end{minipage} 
    \begin{minipage}[b]{.19\textwidth}
        \centering
        \subfloat[\label{fig:stick_cross_pt_2}{}]
        {\includegraphics[width=\textwidth,page=2]{images/stick_cross_pt}}
    \end{minipage}
	\begin{minipage}[b]{.19\textwidth}
        \centering
        \subfloat[\label{fig:stick_cross_pt_3}{}]
        {\includegraphics[width=\textwidth,page=3]{images/stick_cross_pt}}
    \end{minipage}
	\begin{minipage}[b]{.19\textwidth}
        \centering
        \subfloat[\label{fig:stick_cross_pt_4}{}]
        {\includegraphics[width=\textwidth,page=4]{images/stick_cross_pt}}
    \end{minipage}
	\begin{minipage}[b]{.19\textwidth}
        \centering
        \subfloat[\label{fig:stick_cross_pt_5}{}]
        {\includegraphics[width=\textwidth,page=5]{images/stick_cross_pt}}
    \end{minipage}
    \caption{%
    Different configurations used in Lemma~\ref{lem:stick_cross_pt}.}
    \label{fig:stick_cross_pt}
\end{figure}

In the first case, we can assume w.l.o.g.~that $u$ is the vertex
associated with $(v_1,v_2)$, while $u'$ is the one associated with
$(v_s,v_1)$. Hence, there exists an edge, say $f_1$, that crosses
$(v_1,v_2)$ to the right of $(u,u')$, as otherwise we could replace
$(u,u')$ with stick $(v_2,u')$ and reduce the total number of
crossings by one, contradicting the crossing minimality of $G$. Edge
$f_1$ passes through $\mathcal{F}_s$ and also crosses edge
$(v_2,v_3)$ above $(v_1,v_1')$. Similarly, there exists an edge
$f_2$ that crosses $(v_2,v_3)$ below $(v_1,v_1')$, as otherwise
replacing $(v_1,v_1')$ with chord $(v_1,v_3)$ would contradict the
maximality of $G_p$. We proceed by removing edges $(u,u')$ and $f_2$
from $G$ and by replacing them with $(v_3,u)$ and chord $(v_1,v_3)$;
see Figure~\ref{fig:stick_cross_pt_2}. The maximal planar
substructure of the derived graph is larger than $G_p$; a
contradiction.

In the second case, we assume that $u$ is associated with
$(v_1,v_2)$ and $u'$ with $(v_2,v_3)$; see
Figure~\ref{fig:stick_cross_pt_3}. In this scenario, there exists
an edge, say $f$, that crosses $(v_2,v_3)$ below $(v_1,v_1')$, as
otherwise we could replace $(v_1,v_1')$ with chord $(v_1,v_3)$,
contradicting the maximality of $G_p$. If $(v_1,u')$ does not belong
to $G$, then we remove $(u,u')$ from $G$ and replace
it with stick $(v_1,u')$; see Figure~\ref{fig:stick_cross_pt_4}. In
this way, the derived graph has fewer crossings than $G$; a
contradiction. Note that $(v_1,v_1')$ and $(v_1,u')$ cannot be
homotopic (if $v_1' = u'$), as otherwise edge
$(v_1,v_1')$ and $(u,u')$ would not cross in the initial
configuration. Hence, edge $(v_1,u')$ already exists in $G$.
In this case, $f$ is identified with $(v_1,u')$; see
Figure~\ref{fig:stick_cross_pt_5}. But, in this case $f$ is an
uncrossed stick of $\mathcal{F}_s$, contradicting
Lemma~\ref{lem:stick_cross}.\qed
\end{proof}

\begin{lemma}
Let $\fs=\{v_1, v_2, \ldots, v_s\}$, $s \geq 4$ be a non-triangular
face of $G_p$. Then, any stick of $\fs$ is only crossed by some
opposite stick of $\mathcal{F}_s$.
\label{lem:stick_cross_stick}
\end{lemma}
\begin{proof}
By Lemma~\ref{lem:stick_no_far}, each stick of $\fs$ is
short. By Lemma~\ref{lem:stick_cross_once}, each stick of $\fs$ is
crossed exactly once within $\fs$ and this crossing is not with a
middle part due to Lemma~\ref{lem:stick_cross_pt}. For a proof by
contradiction, consider two crossing sticks that are not opposite
and assume w.l.o.g.~that the first stick emanates from vertex $v_1$
(towards vertex $v_1'$) and crosses edge $(v_2,v_3)$, while the
second stick emanates from vertex $v_2$ (towards vertex $v_2'$) and
crosses edge $(v_3,v_4)$; see Figure~\ref{fig:stick_cross_stick_1}.

\begin{figure}[t]
	\centering
    \begin{minipage}[b]{.19\textwidth}
        \centering
        \subfloat[\label{fig:stick_cross_stick_1}{}]
        {\includegraphics[width=\textwidth,page=1]{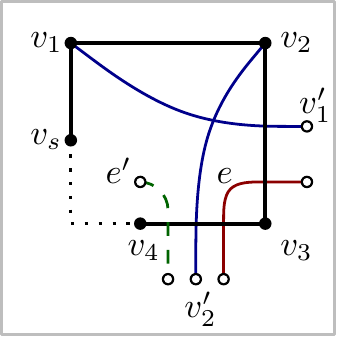}}
    \end{minipage}
    \begin{minipage}[b]{.19\textwidth}
        \centering
        \subfloat[\label{fig:stick_cross_stick_2}{}]
        {\includegraphics[width=\textwidth,page=2]{images/stick_cross_stick}}
    \end{minipage}
	\begin{minipage}[b]{.19\textwidth}
        \centering
        \subfloat[\label{fig:two_stick_1}{}]
        {\includegraphics[width=\textwidth,page=1]{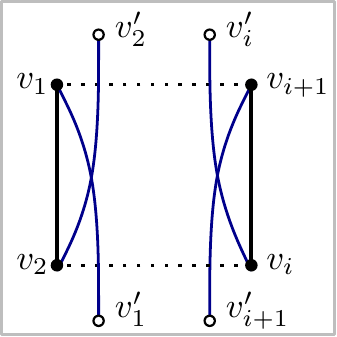}}
    \end{minipage}
	\begin{minipage}[b]{.19\textwidth}
        \centering
        \subfloat[\label{fig:two_stick_2}{}]
        {\includegraphics[width=\textwidth,page=2]{images/two_stick}}
    \end{minipage}
    \caption{%
    Different configurations used in
    (a)-(b)~Lemma~\ref{lem:stick_cross_stick} and
    (c)-(d)~Lemma~\ref{lem:two_stick}.}
    \label{fig:stick_cross_stick_and_two_stick}
\end{figure}

If we can replace $(v_1,v_1')$ with the chord
$(v_1,v_3)$, then the maximal planar substructure of the derived
graph would have more edges than $G_p$; contradicting the maximality
of $G_p$. Thus, there exists an edge, say $e$, that crosses $(v_2,
v_3)$ below $(v_1,v_1')$. By Lemma~\ref{lem:stick_cross_pt}, edge $e$
is passing through $\fs$. Symmetrically, we can prove that there
exists an edge, say $e'$, which crosses $(v_3,v_4)$ right next to
$v_4$, that is, $e'$ defines the closest crossing point to $v_4$
along $(v_3,v_4)$. Note that $e'$ can be either a passing through
edge or a stick of $\fs$.
We proceed by removing from $G$ edges $e'$ and $(v_1,v_1')$ and by
replacing them by the chord $(v_2,v_4)$ and edge $(v_4, v_1')$; see
Figure~\ref{fig:stick_cross_stick_2}. The maximal planar substructure
of the derived graph has more edges than $G_p$ (in the presence of
edge $(v_2,v_4)$), a contradiction.\qed
\end{proof}

\begin{lemma}
Let $\fs=\{v_1, v_2, \ldots, v_s\}$, $s \geq 4$ be a non-triangular
face of $G_p$. Then, $\mathcal{F}_s$ has exactly two sticks.
\label{lem:two_stick}
\end{lemma}
\begin{proof}
By Lemmas~\ref{lem:stick_exist} and~\ref{lem:stick_cross_stick} there
exists at least one pair of opposite crossing sticks. To prove the
uniqueness, assume  that $\fs$ has two pairs of
crossing opposite sticks, say $(v_1,v_1')$, $(v_2,v_2')$ and
$(v_i,v_i')$, $(v_{i+1},v_{i+1}')$, $2 < i < s$; see
Figure~\ref{fig:two_stick_1}. We remove edges
$(v_2,v_2')$ and $(v_i,v_i')$ and replace them by 
$(v_1,v_i)$ and $(v_2, v_{i+1})$; see Figure~\ref{fig:two_stick_2}.
By Lemmas~\ref{lem:pt_no_far} and \ref{lem:stick_no_far}, the newly
introduced edges cannot be involved in crossings. The maximal planar
substructure of the derived graph has more edges than $G_p$ (in the
presence of $(v_1,v_i)$ or $(v_2, v_{i+1})$); a contradiction.\qed
\end{proof}
\noindent We are ready to state the main theorem of this section. 
\begin{theorem}
The planar substructure $G_p$ of a crossing-minimal optimal
$3$-planar graph $G$ is fully triangulated.
\label{thm:triangulated}
\end{theorem}
\begin{proof}
For a proof by contradiction, assume that $G_p$ has a non-triangular
face $\fs=\{v_1,v_2,\ldots,v_s\}$, $s \geq 4$. By
Lemmas~\ref{lem:stick_cross_once}, \ref{lem:stick_cross_stick} and
\ref{lem:two_stick}, face $\fs$ has exactly two opposite
sticks, that cross each other. Assume w.l.o.g.~that these two sticks
emanate from $v_1$ and $v_2$ (towards $v_1'$ and $v_2'$) and exit
$\fs$ by crossing $(v_2,v_3)$ and $(v_1,v_s)$, respectively; recall
that by Lemma~\ref{lem:stick_no_far} all sticks are short; see
Figure~\ref{fig:triangulated_1}.

If we can replace $(v_1,v_1')$ with the chord $(v_1,v_3)$, then the
maximal planar substructure of the derived graph would have more
edges than $G_p$; contradicting the maximality of $G_p$. Thus, there
exists an edge, say $e$, that crosses $(v_2, v_3)$ below
$(v_1,v_1')$. By Lemma~\ref{lem:two_stick}, edge $e$ is passing
through $\fs$. We consider two cases:%
\begin{inparaenum}[(a)]
\item \label{s:c1} edge $(v_2,v_3)$ is only crossed by $e$ and $(v_1,v_1')$,
\item \label{s:c2} there is a third edge, say $e'$, that crosses
$(v_2,v_3)$ (which by Lemma~\ref{lem:two_stick} is also passing
through $\fs$).
\end{inparaenum}

\begin{figure}[t]
	\centering
    \begin{minipage}[b]{.19\textwidth}
        \centering
        \subfloat[\label{fig:triangulated_1}{}]
        {\includegraphics[width=\textwidth,page=1]{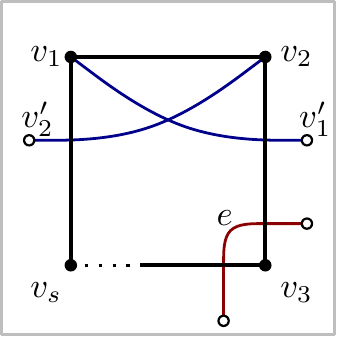}}
    \end{minipage}
    \begin{minipage}[b]{.19\textwidth}
        \centering
        \subfloat[\label{fig:triangulated_2}{}]
        {\includegraphics[width=\textwidth,page=2]{images/triangulated}}
    \end{minipage}
	\begin{minipage}[b]{.19\textwidth}
        \centering
        \subfloat[\label{fig:triangulated_3}{}]
        {\includegraphics[width=\textwidth,page=3]{images/triangulated}}
    \end{minipage}
	\begin{minipage}[b]{.19\textwidth}
        \centering
        \subfloat[\label{fig:triangulated_4}{}]
        {\includegraphics[width=\textwidth,page=4]{images/triangulated}}
    \end{minipage}
    \caption{Different configurations used in Theorem~\ref{thm:triangulated}.}
    \label{fig:triangulated}
\end{figure}

In Case~(\ref*{s:c1}), we can remove from $G$ edges $e$ and
$(v_1,v_1')$, and replace them by $(v_1,v_3)$ and the edge from $v_2$
to the endpoint of $e$ that is below $(v_3,v_4)$; see
Figure~\ref{fig:triangulated_2}. In Case~(\ref{s:c2}), there has to
be a (passing through) edge, say $e''$, surrounding $v_4$ (see
Figure~\ref{fig:triangulated_3}), as otherwise we could replace $e'$
with a stick emanating from $v_4$ towards the endpoint of $e'$ that
is to the right of $(v_2,v_3)$, which contradicts
Lemma~\ref{lem:two_stick}. We proceed by removing from $G$ edges
$e''$ and $(v_1,v_1')$ and by replacing them by $(v_2,v_4)$ and the
edge from $v_2$ to the endpoint of $e''$ that is associated with
$(v_3,v_4)$; see Figure~\ref{fig:triangulated_4}. The maximal planar
substructure of the derived graph has more edges than $G_p$ (in the
presence of $(v_1,v_2)$ in Case~(\ref*{s:c1}) and $(v_2,v_4)$ in
Case~(\ref*{s:c2})), which contradicts the maximality of $G_p$. Since
$G_p$ is connected, there cannot exist a face consisting of only two
vertices.\qed
\end{proof}

\section{Discussion and Conclusion}
\label{sec:conclusions}
This paper establishes a tight upper bound on the number of edges of
non-simple $3$-planar graphs containing no homotopic parallel edges
or self-loops. Our work is towards a complete characterization of
all optimal such graphs. In addition, we believe that our technique
can be used to achieve better bounds for larger values of $k$. We
demonstrate it for the case where $k=4$, where the known bound for
simple graphs is due to Ackerman~\cite{Ackerman15}.

If we could prove that a crossing-minimal optimal $4$-planar graph
$G=(V,E)$ has always a fully triangulated planar substructure
$G_p=(V,E_p)$ (as we proved in Theorem~\ref{thm:triangulated} for the
corresponding $3$-planar ones), then it is not difficult to prove a
tight bound on the number of edges for $4$-planar graphs.
Similar to Lemma~\ref{lem:no_of_sticks}, we can argue that no
triangle of $G_p$ has more than $4$ sticks. Then, we associate each
triangle of $G_p$ with $4$ sticks to a neighboring triangle with at
most $2$ sticks. This would imply $t_4 \leq t_1 + t_2$, where $t_i$
denotes the number of triangles of $G_p$ with exactly $i$ sticks. So,
we would have $|E|-|E_p| = (4t_4 + 3t_3 + 2t_2 + t_1)/2 \leq
3(t_4+t_3+t_2 + t_1)/2 = 3(2n-4)/2 = 3n-6$. Hence, the number of
edges of a $4$-planar graph $G$ is at most $6n - 12$. We
conclude with some open questions.
\begin{itemize}
\item A nice consequence of our work would be the complete
characterization of optimal $3$-planar graphs, as exactly those
graphs that admit drawings where the set of crossing-free edges form
hexagonal faces which contain $8$ additional edges each

\item We also believe that for simple $3$-planar graphs
(i.e., where even non-homotopic parallel edges are not allowed) the
corresponding bound is $5.5n-15$.
  
\item We conjecture that the maximum number of edges of $5$- and
$6$-planar graphs are $\frac{19}{3}n-O(1)$ and $7n-14$, respectively.

\item More generally, is there a closed function on $k$ which
describes the maximum number of edges of a $k$-planar graph for
$k>3$? Recall the general upper bound of $4.1208 \sqrt k n$ by Pach
and T\'oth~\cite{PachT97}.
\end{itemize}
\paragraph*{Acknowledgment:} We thank E. Ackerman for bringing to our
attention \cite{Ackerman15} and~\cite{Pach2006}.

\bibliographystyle{splncs03}
\bibliography{references}

\IfEqCase{\ver}{%
    {arxiv}{\newpage
\appendix
\section*{\LARGE Appendix}

\section{A class of 3-planar graphs with 5.5n--11 edges}
\label{app:bound}
In this section, we demonstrate an infinite class of $3$-planar
graphs with $n$ vertices and exactly $\frac{11n}{2}-11$~edges.

\begin{theorem}
There exist infinitely many $3$-planar graphs with $n$ vertices and
$\frac{11n}{2}-11$ edges.
\label{thm:graphclass}
\end{theorem}
\begin{proof}
Let $n \geq 6$ be a positive integer, such that $n-2$ is divisible
by $4$. Figure~\ref{fig:construction} illustrates an auxiliary plane
graph $H$ with $n$ vertices, $\frac{3(n-2)}{2}$ edges and
$\frac{n-2}{2}$ faces of size $6$. In Figure~\ref{fig:6gon}, we
demonstrate how one can embed $8$ edges in the interior of a face of
size $6$, so that no interior edge is crossed more than three times.
This implies that if we embed this way $8$ edges in every face of
$H$, we will obtain a $3$-planar graph with $n$ vertices and exactly
$\frac{3(n-2)}{2} + 8 \cdot \frac{n-2}{2} = \frac{11n}{2}-11$
edges.\qed
\end{proof}

\begin{figure}[h!]
 \centering
	\begin{minipage}[b]{.19\textwidth}
		\centering
		\subfloat[\label{fig:construction}{}]
		{\includegraphics[width=\textwidth,page=1]{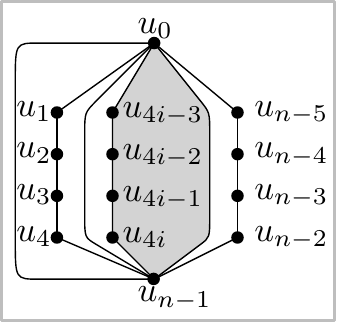}}
	 \end{minipage}
	 \begin{minipage}[b]{.19\textwidth}
		\centering
		\subfloat[\label{fig:6gon}{}]
		{\includegraphics[width=\textwidth,page=2]{images/construction}}
	 \end{minipage}
 \caption{Illustration of 
 (a)~the auxiliary plane graph $H$, and
 (b)~how to embed $8$ edges in a face of size $6$.}
 \label{fig:graphclass} 
\end{figure}

\section{Detailed Proofs from Section~\ref{sec:density}}
\label{app:proofs}

\rephrase{Lemma~\ref{lem:stick_cross}}{ \stickcross }
\begin{proof}
Recall that a stick is the part of an edge from one of its endpoints
towards to the nearest crossing-point with an edge of $G_p$. Hence, a
stick can potentially be further crossed within a face of $G_p$,
i.e., either by another stick or by a middle part of an edge that
passes through this face. Assume to the contrary that there exists a
stick of $\fs$ that is not crossed within $\fs$. W.l.o.g.~let
$(v_1,v_1')$ be the edge containing this stick and assume that
$(v_1,v_1')$ emanates from vertex $v_1$ and leads to vertex $v_1'$ by
crossing the edge $(v_i,v_{i+1})$ of $\fs$. Note that, in general,
$v_1'$ can also be a vertex of $\fs$. For simplicity, we will assume
that $(v_1,v_1')$ is drawn as a vertical line segment with $v_i$ to
the right of $(v_1,v_1')$ and $v_{i+1}$ to the left of $(v_1,v_1')$
as in Figure~\ref{fig:stick_cross_1}. Since $\fs$ is not triangular,
it follows that $i \neq 2$ or $i+1 \neq s$. Assume w.l.o.g.~that $i
\neq 2$.

We initially prove that $i+1=s$. First observe that if we can replace
$(v_1,v_1')$ with the chord $(v_1,v_i)$, then the maximal planar
substructure of the derived graph would have more edges than $G_p$;
contradicting the maximality of $G_p$. We make a remark
here\footnote{This remark will be implicitly used whenever we replace
an existing edge of $G$ with another one (and not explicitly stated
again), throughout this section.}. Edge $(v_1,v_i)$ potentially
exists in $G$ either as part of its planar substructure $G_p$
(because $\fs$ is not necessarily simple) or as part of $G-G_p$. In
the later case, the existence of $(v_1,v_i)$ in $G-G_p$ would deviate
the maximality of $G_p$ (as we showed that $(v_1,v_i)$ can be part of
$G_p$); a contradiction. In the former case, if chord $(v_1,v_i)$
that we introduced is homotopic to an existing copy of $(v_1,v_i)$ in
$G_p$, then $i=2$ must hold; a contradiction.  Hence, there exists an
edge, say $e_1$, that crosses $(v_i, v_{i+1})$ to the right of
$(v_1,v_1')$.

Similarly, if we can replace $e_1$ with the chord $(v_1,v_i)$, then
again the maximal planar substructure of the derived graph would have
more edges than $G_p$; again contradicting the maximality of $G_p$.
Thus, there also exists a second edge, say $e_2$, that crosses $(v_i,
v_{i+1})$ to the right of $e_1$. If $i+1 \neq s$, then a symmetric
argument would imply that $(v_i,v_{i+1})$ has five crossings; a clear
contradiction. Hence, $s=i+1$; see Figure~\ref{fig:stick_cross_2}.

\begin{figure}[t]
	\centering
    \begin{minipage}[b]{.19\textwidth}
        \centering
        \subfloat[\label{fig:stick_cross_1}{}]
        {\includegraphics[width=\textwidth,page=1]{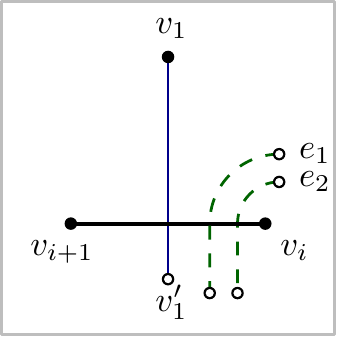}}
    \end{minipage}
    \begin{minipage}[b]{.19\textwidth}
        \centering
        \subfloat[\label{fig:stick_cross_2}{}]
        {\includegraphics[width=\textwidth,page=2]{images/stick_cross}}
    \end{minipage}
	\begin{minipage}[b]{.19\textwidth}
        \centering
        \subfloat[\label{fig:stick_cross_3}{}]
        {\includegraphics[width=\textwidth,page=3]{images/stick_cross}}
    \end{minipage}
	\begin{minipage}[b]{.19\textwidth}
        \centering
        \subfloat[\label{fig:stick_cross_4}{}]
        {\includegraphics[width=\textwidth,page=4]{images/stick_cross}}
    \end{minipage}
    \caption{Different configurations used in Lemma~\ref{lem:stick_cross}.    
    Black edges belong to $G_p$.
    Blue and red edges correspond to sticks and middle parts of $\fs$.
    Green dashed ones are sticks or middle parts of $\fs$.}
    \label{fig:stick_cross}
\end{figure}

We now claim that $e_1$ is not a stick emanating from $v_1$. For a
contradiction, assume that $e_1$ is indeed a stick from
$v_1$. Then, we could replace $e_2$ with the chord $(v_1,v_{s-1})$,
and therefore obtain a graph whose maximal planar substructure has
more edges than $G_p$; contradicting the maximality of $G_p$.
Similarly, $e_2$ is not a stick from $v_1$ (by their definition,
$e_1$ and $e_2$ are not sticks from $v_s$, either).

We now claim that we can remove edges $e_1$, $e_2$ and $(v_1,v_1')$
from $G$ and replace them with the chord $(v_1,v_{s-1})$ and two
additional edges that are both sticks either at $v_1$ or at $v_s$,
as illustrated in Figures~\ref{fig:stick_cross_3} and
\ref{fig:stick_cross_4}, respectively. Indeed, if both
configurations are not possible, then $e_1$ and $e_2$ are homotopic.
Hence, we have obtained a new graph, whose maximal planar
substructure has more edges than $G_p$, which contradicts the
maximality of $G_p$.\qed
\end{proof}

\rephrase{Lemma~\ref{lem:pt_no_far}}{ \ptnofar }
\begin{proof}
For a proof by contradiction, assume that $(u,u')$ is an edge that
defines a middle part of $\fs$ which crosses two non-consecutive
edges of $\fs$, say w.l.o.g.~$(v_1,v_2)$ and $(v_i,v_{i+1})$, where
$i \neq 2$ and $ i+1 \neq s$. As in the proof of
Lemma~\ref{lem:stick_cross}, we will assume for simplicity that
$(u,u')$ is drawn as a vertical line-segment, while $(v_1,v_2)$ and
$(v_i,v_{i+1})$ as horizontal ones, such that $v_1$ and $v_{i+1}$ are
to the left of $(u,u')$ and $v_2$ and $v_i$ to its right. Note that
this might be an oversimplification, if e.g., $v_1$ is identical to
$v_{i+1}$. Clearly, each of $(v_1,v_2)$ and $(v_i,v_{i+1})$ are
crossed by at most two other edges. Let $e_1$, $e_1'$ be the edges
that potentially cross $(v_1,v_2)$ and $e_2$, $e_2'$ the ones that
potentially cross $(v_i,v_{i+1})$. Note that we do not make any
assumption in the order in which these edges cross $(v_1,v_2)$ and
$(v_i,v_{i+1})$ w.r.t.~the edge $(u,u')$; see
Figure~\ref{fig:pt_no_far_1}. Note also that neither $e_1$ nor $e_1'$
can have more than one crossing above $(v_1,v_2)$, as otherwise they
would form sticks of $\fs$ that are not crossed within $\fs$, which
would lead to a contradiction with Lemma~\ref{lem:stick_cross}.
Similarly, $e_2$ and $e_2'$ cannot have more than one crossing below
$(v_i,v_{i+1})$.

First, we consider the case where $(u,u')$ is not involved in
crossings in the interior of $\fs$. Hence, $(u,u')$ can have at most
one additional crossing, either above $(v_1,v_1')$ or below
$(v_i,v_{i+1})$, say w.l.o.g.~below $(v_i,v_{i+1})$. In this case,
we remove edges $(u,u')$, $e_1$, $e_1'$, $e_2$ and $e_2'$ from $G$
and we replace them by the following edges (see also
Figure~\ref{fig:pt_no_far_2}):
\begin{inparaenum}[(a)]
\item the edge from $u$ to $v_{i}$,
\item the edge from $u$ to $v_{i+1}$,
\item the edge from $v_1$ to the endpoint below $(v_i,v_{i+1})$ of
the removed edge that used to cross $(v_i,v_{i+1})$ leftmost,
\item the edge from $v_2$ to the endpoint below $(v_i,v_{i+1})$ of
the removed edge that used to cross $(v_i,v_{i+1})$ rightmost,
\item the edge from $u$ to the endpoint below $(v_i,v_{i+1})$ of
the remaining removed edge that used to cross $(v_i,v_{i+1})$.
\end{inparaenum}
Observe that the maximal planar substructure of the derived graph has
more edges than $G_p$, since it contains edges $(u,v_i)$ and
$(u,v_{i+1})$, instead of edge $(v_1,v_2)$, which contradicts the
maximality of $G_p$.

\begin{figure}[t]
	\centering
    \begin{minipage}[b]{.19\textwidth}
        \centering
        \subfloat[\label{fig:pt_no_far_1}{}]
        {\includegraphics[width=\textwidth,page=1]{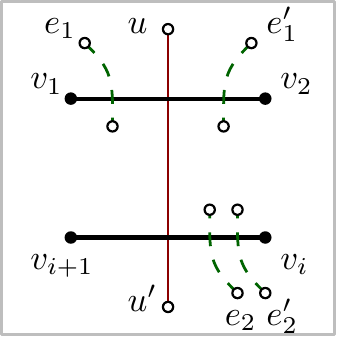}}
    \end{minipage}
    \begin{minipage}[b]{.19\textwidth}
        \centering
        \subfloat[\label{fig:pt_no_far_2}{}]
        {\includegraphics[width=\textwidth,page=2]{images/pt_no_far}}
    \end{minipage}
	\begin{minipage}[b]{.19\textwidth}
        \centering
        \subfloat[\label{fig:pt_no_far_3}{}]
        {\includegraphics[width=\textwidth,page=3]{images/pt_no_far}}
    \end{minipage}
	\begin{minipage}[b]{.19\textwidth}
        \centering
        \subfloat[\label{fig:pt_no_far_4}{}]
        {\includegraphics[width=\textwidth,page=4]{images/pt_no_far}}
    \end{minipage}
    \caption{Different configurations used in Lemma~\ref{lem:pt_no_far}.}
    \label{fig:pt_no_far}
\end{figure}

To complete the proof, it remains to lead to a
contradiction the case where $(u,u')$ is crossed by an edge, say $e$,
within $\fs$; see Figure~\ref{fig:pt_no_far_3}. Observe that edge
$(u,u')$ can be crossed neither above $(v_1,v_1')$ nor below
$(v_i,v_{i+1})$. We proceed to remove $e$, $e_1$, $e_1'$, $e_2$ and
$e_2'$ from $G$ and we replace them by the edges $(v_2,v_{i+1})$,
$(u,v_{i+1})$, $(u,v_i)$, $(u',v_1)$ and $(u',v_2)$, respectively;
see Figure~\ref{fig:pt_no_far_4}. The planar substructure of the
derived graph has more edges than $G_p$; a contradiction.\qed
\end{proof} 

\rephrase{Lemma~\ref{lem:stick_exist}}{ \stickexist }
\begin{proof}
For a proof by contradiction, assume that $\fs$ has no sticks. By
Lemma~\ref{lem:uncrossed_edges}, it follows that there exist at least
two incident edges of $\overline{br}(\fs)$ that are crossed by
passing through edges of $\fs$, say w.l.o.g.~$(v_s,v_1)$ and
$(v_1,v_2)$. Note that these two edges are not bridges of $\fs$.  We
remove all passing through edges of $\fs$ and we add several new
edges in $\fs$; see also Figure~\ref{fig:stick_exist_1}.
As in the proof of Lemma~\ref{lem:uncrossed_edges}, we introduce
$s-3$ edges $\{(v_1,v_i):~2<i<s\}$ that lie completely in the
interior of $\fs$. Let $e_i=(v_i,v_{i+1})$, $2 < i < s$ be an edge of
$\overline{br}(\fs)$, other than $(v_s,v_1)$ and $(v_1,v_2)$, that
was crossed by a passing through edge of $\fs$. Let also $u_i$ be the
vertex associated with this particular edge. Then, we can introduce
edge $(v_1,u_i)$ in $G$ by maintaining $3$-planarity as follows: we
draw this edge starting from $v_1$ and between edges $(v_1,v_i)$ and
$(v_1,v_{i+1})$, towards the crossing point along $e_i$ and then we
follow the part of the passing through edge associated with $e_i$
towards $u_i$. Hence, potential parallel edges are not homotopic. In
the same way, we introduce two more edges starting from $v_3$ and
$v_{s-1}$ towards to the two vertices associated with $(v_1,v_2)$ and
$(v_1,v_s)$, respectively (recall that both $(v_1,v_2)$ and
$(v_1,v_s)$ were initially involved in crossings).

Since $\widehat{s_b}$ is the number of edges of $\overline{br}(\fs)$
that initially were not crossed by any passing through edge of $\fs$,
in total we have introduced $s-3+\overline{s_b}-\widehat{s_b}$
edges (recall that $s=2s_b+\overline{s_b}$). Since every edge of
$\overline{br}(\fs)$ can be crossed at most three times and each
passing through edge of $\fs$ crosses two edges of
$\overline{br}(\fs)$, it follows that initially we removed at most
$\lfloor \frac{3}{2}(\overline{s_b}-\widehat{s_b}) \rfloor$ edges.
This implies that as long as $s+\widehat{s_b} + 2s_b \geq 6$, the
resulting graph is larger or of equal size as $G$ but with larger
planar substructure. In the case where $s+\widehat{s_b} +2s_b = 5$
(that is, $s=5$ and $\widehat{s_b} = s_b=0$ or $s=4$,
$\widehat{s_b}=1$ and $s_b=0$), the resulting graph is again of equal
size as $G$ but with larger planar~substructure. Both cases, of
course, contradict either the optimality of $G$ of the maximality of
$G_p$.

To complete the proof of this lemma, it remains to lead to a
contradiction the case, where $s+\widehat{s_b}+2s_b=4$. Since $\fs$
is not triangular, $s=4$ and $\widehat{s_b}=s_b=0$ follows.  Recall
that in this case $\fs$ initially consisted of four edges, each of
which was crossed exactly three times by some passing through edges
(out of six in total). Let $R_i$ be the set of all possible vertices
that can be associated with $(v_i,v_{i+1})$, $i=1,\ldots,4$. Clearly,
$1 \leq |R_i| \leq 3$. Let also $u_i$ be a vertex of $R_i$.
By Lemma~\ref{lem:pt_no_far} it follows that all passing through
edges with an endpoint in $R_i$ have their other endpoint in
$R_{i+1}$ or in $R_{i-1}$. Suppose first, for some $i=1,\ldots,4$,
that all passing through edges with an endpoint in $R_i$ have their
other endpoint in $R_{i+1}$ and not in $R_{i-1}$. In this scenario,
however, it is clear that edge $(v_i,v_{i+2})$ can be safely added to
$G$ without destroying its $3$-planarity, which of course contradicts
the optimality of $G$ (see Figure~\ref{fig:stick_exist_4}). Hence,
for every $i=1,\ldots,4$ there exists a passing through edge with an
endpoint in $R_i$ and its other endpoint in $R_{i+1}$.
To cope with this case, we replace all passing through edges of $\fs$
with the edges of the configuration illustrated either in
Figure~\ref{fig:stick_exist_2} or \ref{fig:stick_exist_3}. Both
configurations are suitable in this case. Additionally, the presence
of $(v_2,v_4)$ or $(v_1,v_3)$, respectively, leads to a contradiction
the maximality of the planar substructure. Observe that edges
$(u_1,u_3)$ and $(u_2,u_4)$ are both involved in three crossings
each. This implies that both configurations might be forbidden (due
to $3$-planarity), in the case where all passing through edges that
initially emanated, say w.l.o.g.~from each vertex of $R_1$ and each
vertex of $R_2$, had crossings outside $\fs$. This implies, however,
that initially there was no passing through edge of $\fs$ from a
vertex of $R_1$ to a vertex of $R_2$ (as such an edge would have four
crossings); a contradiction.\qed
\end{proof}

\begin{figure}[t]
	\centering
    \begin{minipage}[b]{.19\textwidth}
        \centering
        \subfloat[\label{fig:stick_exist_1}{}]
        {\includegraphics[width=\textwidth,page=1]{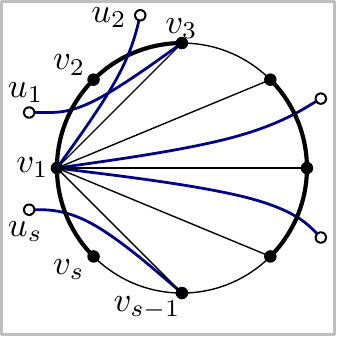}}
    \end{minipage}
	\begin{minipage}[b]{.19\textwidth}
        \centering
        \subfloat[\label{fig:stick_exist_4}{}]
        {\includegraphics[width=\textwidth,page=4]{images/stick_exist}}
    \end{minipage}
    \begin{minipage}[b]{.19\textwidth}
        \centering
        \subfloat[\label{fig:stick_exist_2}{}]
        {\includegraphics[width=\textwidth,page=2]{images/stick_exist}}
    \end{minipage}
	\begin{minipage}[b]{.19\textwidth}
        \centering
        \subfloat[\label{fig:stick_exist_3}{}]
        {\includegraphics[width=\textwidth,page=3]{images/stick_exist}}
    \end{minipage}	
    \caption{%
    Different configurations used in Lemma~\ref{lem:stick_exist}.}
    \label{fig:stick_exist}
\end{figure}

\rephrase{Lemma~\ref{lem:stick_cross_once}}{ \stickcrossonce }
\begin{proof}
By Lemma~\ref{lem:stick_cross}, each stick of $\fs$ is crossed at
least once within $\fs$. For a proof by contradiction, assume that
there exists a stick of $\fs$ that is crossed twice within $\fs$ (by
edges $e_1$ and $e_2$; see Figure~\ref{fig:stick_cross_once_1}).
W.l.o.g.~let $(v_1,v_1')$ be the edge containing this stick and
assume that $(v_1,v_1')$ emanates from vertex $v_1$ and leads to
vertex $v_1'$ by crossing the edge $(v_2,v_3)$ of $\fs$, that is,
$(v_1,v_1')$ forms a right stick of $\fs$ (recall that by
Lemma~\ref{lem:stick_no_far}, each stick of $\fs$ is short).

First, we show that $e_1$ and $e_2$ cannot cross in $\fs$.
Assume to the contrary that this is not the case, namely, $e_1$
crosses $e_2$ in $\fs$; see Figure~\ref{fig:stick_cross_once_1}.
Since $e$, $e_1$ and $e_2$ mutually cross in $\fs$, both $e_1$ and
$e_2$ have two crossings within $\fs$. It follows that neither $e_1$
nor $e_2$ passes through $\fs$, or equivalently, that both $e_1$ and
$e_2$ form sticks of $\fs$. This, however, contradicts
Lemma~\ref{lem:stick_no_three}, as $e$, $e_1$ and $e_2$ define three
mutually crossing sticks of $\fs$. Before we continue, we make two
useful remarks:%

\begin{enumerate}[{R}.1.] 
\item \label{rem:1} Let $\mathcal{F}'$ be the face of $G_p$ that
shares edge $(v_2,v_3)$ with $\fs$. Since $e$ has already three
crossings within $\fs$, it follows that $v'_1$ is a vertex of
$\mathcal{F}'$. For face $\mathcal{F}'$, edge $e$ forms an uncrossed
stick. Hence, $\mathcal{F}'$ is triangular and $\mathcal{F}'
\neq \fs$ (refer to the gray-colored face of
Figure~\ref{fig:stick_cross_once_1}).

\item \label{rem:2} Assume that either $e_1$ or $e_2$, say
w.l.o.g.~$e_1$, is passing through $\fs$. By
Lemma~\ref{lem:pt_no_far}, it follows that $e_1$ is crossing either
$(v_1,v_s)$ or $(v_2,v_3)$ of $\fs$. We claim that $e_1$ cannot cross
$(v_2,v_3)$. For a proof by contradiction, assume that this is not
the case. If $e_1$ passes through $\mathcal{F}'$, then $e_1$ would
have at least four crossings in the drawing of $G$; a contradiction.
So, $v_1'$ is an endpoint of $e_1$. However, in this case, $e_1$ and
$(v_1,v_1')$ would not cross in the initial drawing of $G$; a
contradiction. Hence, $e_1$ is crossing $(v_1,v_s)$ of $\fs$. Let
w.l.o.g.~$e_1=(u,v)$. Arguing similarly with Remark~R.\ref*{rem:1},
we can show that  edges $(v_1,v_s)$ and $(v_1,v_2)$ belong to two
triangular faces in $G_p$ with $u$ and $w$ as third vertex,
respectively (see Figure~\ref{fig:stick_cross_once_2}). Hence, $e_2$
cannot simultaneously pass through $\fs$. We distinguish two
cases depending on whether $e_1$ passes through $\fs$ or not.
\end{enumerate}

\begin{description}
\item[-] \emph{Edge $e_1$ passes through $\fs$}; see
Figure~\ref{fig:stick_cross_once_2}.
By $3$-planarity, there are at most two more edges, say $f_1$,
$f_2$, that cross edge $(v_1,v_s)$ and at most two more edges, say
$g_1$, $g_2$, that cross $(v_2,v_3)$. We remove these edges from $G$
as well as edges $e_1$ and $e_2$, i.e., a total of at most $6$
edges, and we replace them with the edges $(u,v_2)$, $(u,v_1')$,
$(u,v_3)$, $(v_1',v_s)$ $(v_1,v_3)$ and $(v_3,v_s)$; see
Figure~\ref{fig:stick_cross_once_3}. If $s>4$ or one among $f_1$,
$f_2$, $g_1$ and $g_2$ is not present in $G$, then the derived graph
has at least as many edges as $G$ but its maximal planar
substructure has two more edges, i.e., $(v_1,v_3)$ and $(v_3,v_s)$,
contradicting the maximality of $G_p$.

\begin{figure}[t] 
    \centering
	\begin{minipage}[b]{.19\textwidth}
        \centering
        \subfloat[\label{fig:stick_cross_once_1}{}]
        {\includegraphics[width=\textwidth,page=1]{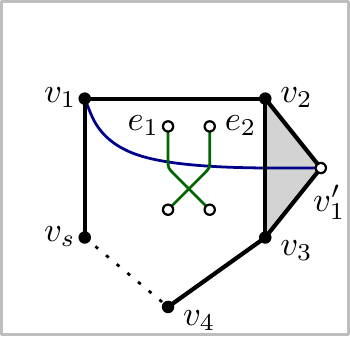}}
    \end{minipage}
	\begin{minipage}[b]{.19\textwidth}
        \centering
        \subfloat[\label{fig:stick_cross_once_2}{}]
        {\includegraphics[width=\textwidth,page=2]{images/stick_cross_once}}
    \end{minipage}
    \begin{minipage}[b]{.19\textwidth}
        \centering
        \subfloat[\label{fig:stick_cross_once_3}{}]
        {\includegraphics[width=\textwidth,page=3]{images/stick_cross_once}}
    \end{minipage}
	\begin{minipage}[b]{.19\textwidth}
        \centering
        \subfloat[\label{fig:stick_cross_once_4}{}]
        {\includegraphics[width=\textwidth,page=4]{images/stick_cross_once}}
    \end{minipage}
	\begin{minipage}[b]{.19\textwidth}
        \centering
        \subfloat[\label{fig:stick_cross_once_5}{}]
        {\includegraphics[width=\textwidth,page=5]{images/stick_cross_once}}
    \end{minipage}
    \caption{%
    Different configurations used in Lemma~\ref{lem:stick_cross_once}: 
    The case where edge $e_1$ passes through $\fs$.}
    \label{fig:stick_cross_once_part_1}
\end{figure}

Consider now the case where edges $f_1$, $f_2$, $g_1$ and $g_2$ are
present in $G$ and $s=4$. In this case, edge $(v_3,v_s)$ exists in
$G$. By $3$-planarity, $f_1$ and $f_2$ cross $(v_1,v_4)$ below
$e_1$. Also, at least one of $g_1$ and $g_2$, say w.l.o.g.~$g_1$,
crosses $(v_2,v_3)$ bellow $(v_1,v_1')$, otherwise we could replace
$(v_1,v_1')$ with chord $(v_1,v_3)$, contradicting the maximality of
$G_p$. The second edge $g_2$ may cross $(v_2,v_3)$ either above
$(v_1,v_3)$ or below $(v_1,v_3)$; see
Figure~\ref{fig:stick_cross_once_2}.

We claim that $e_2$ and $(v_3,v_4)$ cannot cross. For a proof by
contradiction, assume that $e_2$ and $(v_3,v_4)$ cross. By
$3$-planarity, at most two of edges $f_1$, $f_2$ and $g_1$ can cross
$(v_3,v_4)$. Thus, at least one of them is a stick crossing $e_2$.
Since $e_2$ has already three crossings, it must be a stick of
$v_2$. This implies that exactly two of $f_1$, $f_2$ and $g_1$ cross
$(v_3,v_4)$. On the other hand, $g_2$ can cross neither $e_2$ nor
$(v_3,v_4)$. Hence, $g_2$ cannot exist; a contradiction.

Since $e_2$ and $(v_3,v_4)$ cannot cross, edge $e_2$ forms a stick
emanating either from $v_3$ or from $v_4$. In the later case, $e_2$
must cross $f_1$ and $f_2$, and therefore has at least four
crossings (as it also crosses $(v_1,v_1')$ and an edge of $\fs$ to
exit $\fs$); a contradiction.

From the above, it follows that $e_2$ forms a stick of $v_3$; see
Figure~\ref{fig:stick_cross_once_4}. In this case, $e_2$ crosses
with $(v_1,v_2)$, $(v_1,v_1')$ and $g_1$ (which crosses $(v_2,v_3)$
below $(v_1,v_1')$). Since $e_2$ has already three crossings, it
follows that $g_2$ crosses $(v_2,v_3)$ above $(v_1,v_1')$ and passes
through $\fs$. Also, $g_1$ cannot be a stick of $v_4$, as otherwise
it would cross with both $f_1$ and $f_2$ having more than three
crossings. So, $g_1$ crosses $(v_3,v_4)$ and passes through $\fs$.
Similarly to Remark~R.\ref*{rem:2}, we can show that $g_1$ joins
vertex $v'_1$ with a vertex, say $w'$, so that $w'$, $v_3$ and $v_4$
form a triangular face of $G_p$. It follows that vertices $v_1$,
$w$, $v_2$, $v'_1$, $v_3$, $w'$, $v_4$ and $u$ form an octagon in
$G_p$ with $4$ edges of $G_p$ in its interior and a total of $7$
more edges of $G-G_p$ that either lie entirely in the octagon or
pass through the octagon. We remove these $11$ edges from $G$ and
replace them with the corresponding ones of
Figure~\ref{fig:stick_cross_once_5} (which lie completely in the
interior of the octagon). In the derived graph, the octagon has
still a total of $11$ edges. However, $5$ of them belong to its
maximal planar substructure; a contradiction to the maximality of
$G_p$.

\item[-] \emph{Edge $e_1$ is a stick of $\fs$}. In this case, both
$e_1$ and $e_2$ form sticks of $\fs$ (by Remark~R.\ref*{rem:2}). By
Lemma~\ref{lem:stick_no_far} and by the fact that $e_1$ and $e_2$
cross $(v_1,v_1')$, \underline{$e_1$ and $e_2$ emanate from $v_2$,
$v_3$ or $v_s$.}

First, we will prove that \underline{neither $e_1$ nor $e_2$ forms a
stick of $v_3$}. For a proof by contradiction, assume that $e_2$
forms a stick of $v_3$; see Figure~\ref{fig:stick_cross_once_6}.
Since $e_1$ and $e_2$ do not cross, $e_1$ forms stick of either $v_3$
or $v_s$. In the former case, however, we can add edge $(v_1,v_3)$ to
$G$, contradicting its optimality. Therefore, edge $e_1$ forms a
stick of $v_s$. Edge $g_1$ crosses $(v_2,v_3)$ bellow $(v_1,v_1')$,
as otherwise we could replace $(v_1,v_1')$ with chord $(v_1,v_3)$
contradicting the maximality of $G_p$. It follows that $g_1$ also
crosses $e_2$. This implies that $g_1$ is a stick of $\fs$. Since
$e_2$ has three crossings, it follows that $e_2$ joins $v_3$ with a
vertex, say $v_3'$, so that $v_1$, $v_2$ and $v'_3$ form a triangular
face of $G_p$. By $3$-planarity, the third edge $g_2$ that
potentially crosses $(v_2,v_3)$ lies above $(v_1,v_1')$ and passes
through $\fs$. Also by $3$-planarity, there exists at most one other
edge $f_1$ that crosses $e_1$ and is a stick of $\fs$ (as shown in
the first part of the proof). Consider now the ``hexagon'' defined by
$v_s$ $v_1$, $v'_3$, $v_2$, $v'_1$ and $v_3$. It contains two or
three edges of $G_p$ (depending on whether $s>4$ or $s=4$,
respectively) and at most $5$ other edges. We remove them from $G$
and replace them with the corresponding ones of
Figure~\ref{fig:stick_cross_once_7}. The derived graph has at least
as many edges as $G$, but its planar substructure is larger than
$G_p$ (due to chord $(v_1,v_3)$); a contradiction to the maximality
of $G_p$. So, $e_1$ and $e_2$ are sticks of $v_2$ or~$v_s$.

\begin{figure}[t]
	\centering
	\begin{minipage}[b]{.19\textwidth}
        \centering
        \subfloat[\label{fig:stick_cross_once_6}{}]
        {\includegraphics[width=\textwidth,page=6]{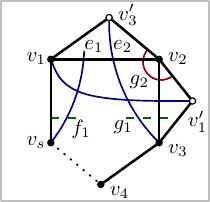}}
    \end{minipage}
    \begin{minipage}[b]{.19\textwidth}
        \centering
        \subfloat[\label{fig:stick_cross_once_7}{}]
        {\includegraphics[width=\textwidth,page=7]{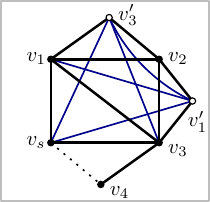}}
    \end{minipage}
    \begin{minipage}[b]{.19\textwidth}
        \centering
        \subfloat[\label{fig:stick_cross_once_8}{}]
        {\includegraphics[width=\textwidth,page=8]{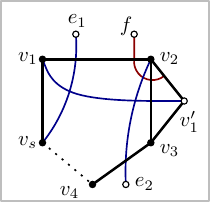}}
    \end{minipage}	
    \begin{minipage}[b]{.19\textwidth} 
        \centering
        \subfloat[\label{fig:stick_cross_once_9}{}]
        {\includegraphics[width=\textwidth,page=9]{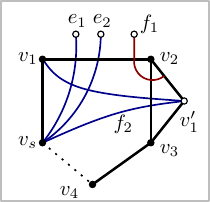}}
    \end{minipage}	
    \begin{minipage}[b]{.19\textwidth}
        \centering
        \subfloat[\label{fig:stick_cross_once_10}{}]
        {\includegraphics[width=\textwidth,page=10]{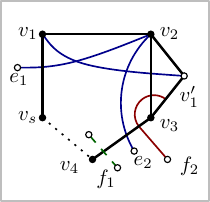}}
    \end{minipage}	
    \begin{minipage}[b]{.19\textwidth}
        \centering
        \subfloat[\label{fig:stick_cross_once_11}{}]
        {\includegraphics[width=\textwidth,page=11]{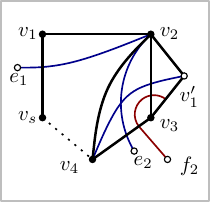}}
    \end{minipage}
    \begin{minipage}[b]{.19\textwidth}
        \centering
        \subfloat[\label{fig:stick_cross_once_12}{}]
        {\includegraphics[width=\textwidth,page=12]{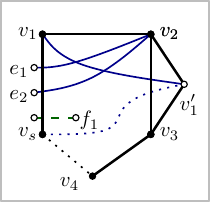}}
    \end{minipage}	
    \begin{minipage}[b]{.19\textwidth}
        \centering
        \subfloat[\label{fig:stick_cross_once_13}{}]
        {\includegraphics[width=\textwidth,page=13]{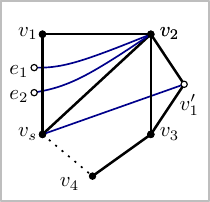}}
    \end{minipage}	
    \begin{minipage}[b]{.19\textwidth}
        \centering
        \subfloat[\label{fig:stick_cross_once_14}{}]
        {\includegraphics[width=\textwidth,page=14]{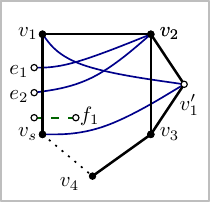}}
    \end{minipage}	
    \begin{minipage}[b]{.19\textwidth}
        \centering
        \subfloat[\label{fig:stick_cross_once_15}{}]
        {\includegraphics[width=\textwidth,page=15]{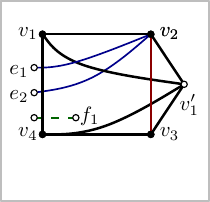}}
    \end{minipage}
    \caption{%
	Different configurations used in Lemma~\ref{lem:stick_cross_once}: 
    The case where edge $e_1$ forms a stick of $\fs$.}
    \label{fig:stick_cross_once_part_2} 
\end{figure}

Next, we will prove that \underline{$e_1$ and $e_2$ emanate from the
same vertex of $\fs$}. For a proof by contradiction, assume that
$e_1$ is a stick of $v_s$ and $e_2$ is a stick of $v_2$; see
Figure~\ref{fig:stick_cross_once_8}. By
Lemma~\ref{lem:stick_no_three}, edge $e_2$ crosses edge $(v_3,v_4)$
of $\fs$. Now, there exists an edge $f$ that crosses $(v_1,v_2)$ to
the right of $e_1$, otherwise we could replace $e_1$ with chord
$(v_s,v_2)$ contradicting the maximality of $G_p$. This edge also
crosses $e_2$ and $(v_2,v_3)$, that is, $f$ passes through $\fs$.
Then, $e_2$ is a stick of $\fs$ that is crossed twice: by a stick
and a passing through edge. This case however cannot occur, since it
is covered by the first case of the lemma. So, $e_1$ and $e_2$ are
sticks of the same vertex of $\fs$. 

Next, we will prove that \underline{$e_1$ and $e_2$ do not form
sticks of $v_s$}; see Figure~\ref{fig:stick_cross_once_9}. As before,
there exists an edge $f_1$ that passes through $\fs$ and crosses
$(v_1,v_2)$ to the right of $e_2$ and $(v_2,v_3)$ above $(v_1,v_1')$,
as otherwise we could replace $e_2$ with chord $(v_s,v_2)$
contradicting the maximality of $G_p$. Similarly, there exists an
edge $f_2$ that crosses $(v_2,v_3)$ bellow $(v_1,v_1')$, as otherwise
we could replace $(v_1,v_1')$ with chord $(v_1,v_3)$ and lead to a
contradiction the maximality of $G_p$. We claim that $f_2$ is an edge
connecting $v_s$ with $v_1'$. First, we make the following
observation. Suppose that there is an edge that crosses $e_1$ and
$e_2$ within $\fs$. By Remark~\ref*{rem:1}, $e_1$ and $e_2$ are
homotopic; a contradiction. Therefore, no further edge
crosses $e_1$ and $e_2$. Now, if $f_2$ is not an edge
connecting $v_s$ with $v_1'$, then we can replace $(v_1,v_1')$ with
the edge $(v_s,v_1')$ and reduce the total number of crossings of $G$
by two, which of course contradicts the crossing minimality of $G$.
If $s > 4$, clearly we can add edge $(v_3,v_s)$ to $G$ and contradict
its optimality. Therefore, $s=4$ holds. In this case, $f_2$ is a
stick of $\fs$. Hence, by Lemma~\ref{lem:stick_cross} $f_2$ must be
crossed at least once within $\fs$, which is not possible in the
absence of chord $(v_1,v_3)$ because of the $3$-planarity.

It remains to prove that \underline{$e_1$ and $e_2$ do not form
sticks of $v_2$}. Assuming that $e_2$ crosses $(v_1,v_1')$ rightmost
(among $e_1$ and $e_2$), we consider two cases: $e_2$ forms a%
\begin{inparaenum}[(i)]
\item \label{st:1} right or
\item \label{st:2} left stick of $\fs$.
\end{inparaenum}

Case~(\ref*{st:1}) is illustrated in
Figure~\ref{fig:stick_cross_once_10}. In this case, there exists an
edge $f_1$ that crosses $(v_3,v_4)$ to the left of $e_2$, as
otherwise we could replace $e_2$ with chord $(v_2,v_4)$ contradicting
the maximality of $G_p$. Note that if $f_1= e_1$, then $e_1$ can be
replaced with chord $(v_2,v_4)$, again leading to a contradiction the
maximality of $G_p$. Analogously, there exists an edge $f_2$ that
crosses $(v_2,v_3)$ bellow $(v_1,v_1')$, as otherwise we could
replace $(v_1,v_1')$ with chord $(v_1,v_3)$, which would contradict
the maximality of $G_p$. By $3$-planarity, edge $f_2$ cannot cross
$e_2$. Hence, $f_2$ passes through $\fs$ and crosses $(v_3,v_4)$ to
the right of $e_2$. This implies that $e_1$ is a left stick and
crosses $(v_1,v_s)$. We proceed by removing $(v_1,v_1')$ and $f_1$
from $G$ and by replacing them with edge $(v_4,v'_1)$ and chord
$(v_2,v_4)$; see Figure~\ref{fig:stick_cross_once_11}. Note that this
replacement is legal, since we can show (as in the case where $e_1$
and $e_2$ do not form sticks of $v_s$) that $(v_2,v_3)$ is not
involved in any other crossing. The maximal planar substructure of
the derived graph is larger than $G_p$; a contradiction.

Case~(\ref*{st:2}) is illustrated in
Figure~\ref{fig:stick_cross_once_12}. In this case, both $e_1$ and
$e_2$ form left sticks of $v_2$. In addition, there exists an edge
$f_1$ that crosses $(v_1,v_s)$ bellow $e_2$, as otherwise we could
replace $e_2$ with chord $(v_2,v_s)$ contradicting the maximality of
$G_p$. In the absence of $(v_s,v'_1)$, we remove $(v_1,v_1')$ and
$f_1$ from $G$ and we replace them with $(v_s,v'_1)$ and chord
$(v_2,v_s)$. The maximal planar substructure of the derived graph has
more edges than  $G_p$, which contradicts its maximality. Hence,
$(v_s,v'_1)$ belongs to $G$; see
Figure~\ref{fig:stick_cross_once_14}. If $s>4$, then
$(v_s,v_1')$ forms a far stick of $\fs$, contradicting
Lemma~\ref{lem:stick_no_far}. Hence $s=4$. In this case,
we can remove $(v_2,v_3)$ from $G_p$ and add edges $(v_s,v_1')$ and
$(v_1,v'_1)$ to it, which again contradicts the maximality
of $G_p$; see Figure~\ref{fig:stick_cross_once_15}.\qed%
\end{description}
\end{proof}}%
}
\end{document}